\documentclass[a4paper,final,12pt]{article}
\usepackage[table]{xcolor}
\usepackage{enumitem,amsmath,amsfonts,amsthm,parskip,graphicx,url,amssymb,latexsym,bbm,xr,tikz,abstract,parskip,float,fontenc}
\usepackage{tikz, pgfplots, pst-plot}
\usepackage{centernot}
\usepackage[justification=centering]{caption}
\usepackage{subcaption}
\usepackage{mathtools}
 \usepackage{mathabx}
\usepackage[margin=1.25in]{geometry}
\usepackage[onehalfspacing]{setspace}
\definecolor{MyDarkBlue}{rgb}{0,0.08,0.45}
\definecolor{cites}{HTML}{324b13}
\definecolor{links}{HTML}{1a663b}
\definecolor{MyLightMagenta}{cmyk}{0.1,0.8,0,0.1}
\definecolor{sblue}{HTML}{0049A9}
\definecolor{scyan}{HTML}{CBEAFC}
\definecolor{sred}{HTML}{B5595C}
\definecolor{sgreen}{HTML}{609B57}
\definecolor{spink}{HTML}{FFB0FF}
\usepackage{clipboard}	
\newclipboard{PaperClipboard}	
\usepackage{xr} 
\usepackage[colorlinks,linkcolor=links,citecolor=cites,urlcolor=MyDarkBlue]{hyperref}
\usepackage{cleveref}
\usepackage[section]{placeins}
\usepackage[english]{babel}
\usepackage{natbib}
\usepackage[bottom,splitrule]{footmisc}

\makeatletter
\def\thm@space@setup{%
  \thm@preskip=\parskip \thm@postskip=0pt
}
\makeatother

\newtheorem*{AxiomT}{Axiom $\tau$}
\newtheorem*{AxiomR}{Axiom $\rho$}

\newtheorem{Proposition}{Proposition}
\theoremstyle{definition}\newtheorem{Example}{Example}
\newtheorem{Definition}{Definition}

\newtheorem{Lemma}{Lemma}
\newtheorem{theorem}{Theorem}

\makeatletter
\renewcommand\section{\@startsection {section}{1}{\z@}%
                                   {-1.5ex \@plus -1ex \@minus -.2ex}%
                                   {1.5ex \@plus.2ex}%
                                   {\normalfont\scshape}}
\renewcommand\subsection{\@startsection{subsection}{2}{\z@}%
                                     {-1.5ex\@plus -1ex \@minus -.2ex}%
                                     {0.25ex \@plus .2ex}%
                                     {\normalfont\itshape}}
\makeatother
\newcommand{\weakb}{\ensuremath{\triangleright}}
\newcommand{\strictb}{\ensuremath{\triangleright}}
\newcommand{\nstrictb}{\ensuremath{\ntriangleright}}
\title{Disentangling Revealed Preference From Rationalization by a Preference\thanks{I thank the Editor, Joseph Harrington, and an anonymous referee for feedback that greatly improved this paper.}}
\author{Pablo Schenone\thanks{Department of Economics, Fordham University. Mailing address: Fifth Floor, Dealy Hall, 441 E Fordham Rd, Bronx, NY 10458. Phone number: (718) 817-4048. E-mail: \href{mailto:pschenone@fordham.edu}{\texttt{pschenone@fordham.edu}}}}%

\begin{document}
\pagenumbering{gobble}
\maketitle
\begin{abstract}
The weak axiom of revealed preference (WARP) ensures that the revealed preference (i) is a preference relation (i.e., it is complete and transitive) \emph{and} (ii) rationalizes the choices. However, when WARP fails, either one of these two properties is violated, but it is unclear which one it is. We  provide an alternative characterization of WARP by showing that WARP is equivalent to the conjunction of two axioms each of which \emph{separately} guarantees (i) and (ii). 
\end{abstract}
\normalsize
\newpage
\clearpage
\pagenumbering{arabic}
\section{Introduction}\label{section:intro}


\label{page-intro-rp} A foundational result in economic theory is that a choice correspondence satisfies the Weak Axiom of Revealed Preference (henceforth, WARP) if and only if there exists a complete and transitive ranking of alternatives--a preference relation--such that the choice correspondence maps any choice set into the subset of undominated alternatives--i.e., it rationalizes choices (see, e.g.,  \cite{arrow1959rational} and \cite{sen1971choice}). Furthermore, the preference relation that rationalizes the choice correspondence is the \emph{revealed preference}, defined as follows: an alternative $x$ is (strictly) revealed preferred to an alternative $y$ if $x$ is (uniquely) chosen from the binary set $\{x,y\}$.

In the above result, WARP performs two tasks. First, it guarantees that the revealed preference is indeed a preference relation (see Definitions \ref{def:strict preference} and \ref{def:weak preference} for formal statements). Second, it guarantees that the revealed preference rationalizes the choice correspondence. However, as the examples below show, these two properties are independent: the revealed preference may be a preference relation--even if it does not rationalize choice--and choice may be rationalized by the revealed preference--even if the revealed preference is not a preference relation.

\begin{Example}[Revealed preference is a preference relation but does not rationalize choice]\label{Example:example1}
 A decision maker chooses between four different fruits at their local grocery shop: apples ($a$), bananas ($b$), kiwis ($k$), and dates ($d$). The decision maker makes choices as follows: 
\begin{enumerate}[label=(\arabic*)]
\item\label{itm:apple-best} From sets $\{a,b\}$, $\{a,k\}$, and $\{a,d\}$, the choice is $\{a\}$,
\item\label{itm:banana-kiwi} From set $\{b,k\}$, the choice is $\{b, k\}$,
\item\label{itm:kiwis-beat-dates} From set $\{k, d\}$, the choice is $\{k\}$,
\item\label{itm:banana-beat-dates} From set $\{b, d\}$, the choice is $\{b\}$,
\item\label{itm:full} From the full set of alternatives, $\{a,b,k,d\}$,  the choice is $\{a,d\}$
\end{enumerate}

\label{page-ex-1-rp}These choices imply the following. First, the revealed preference, {as defined in the opening paragraph}, is complete and transitive: {from (1) we conclude that} apples are {strictly revealed preferred to bananas, kiwis and dates}; {from (2) we conclude that} neither bananas nor kiwis are strictly revealed preferred to each other, and {from (3) and (4) we conclude that }dates are the least preferred alternative. Second, from (5) we conclude that the revealed preference does \emph{not} rationalize the decision maker's choices, as choosing dates out of the full set of alternatives means choosing a dominated alternative.\hfill$\lozenge$ 

\end{Example}

\begin{Example}[Revealed preference rationalizes choices but is not a preference]\label{Example:example2}
\indent A decision maker  chooses between three vacation spots: Amsterdam ($a$), Brussels ($b$), or Korea ($k$).
\begin{enumerate}[label=(\arabic*)]
\item From set $\{a,b\}$, the choice is $\{a\}$,
\item From set $\{b,k\}$, the choice is $\{b, k\}$,
\item From set $\{a,k\}$, the choice is $\{a,k\}$,
\item From the full set of alternatives, $\{a,b,k\}$,  the choice is $\{a,k\}$.
\end{enumerate}

This example is the opposite of \autoref{Example:example1}. First, the revealed preference is intransitive: {from (1) we conclude that } Amsterdam is strictly revealed preferred to Brussels, {from (2) we conclude that} neither Brussels nor Korea are strictly revealed preferred to each other, but {(3) implies that} Amsterdam is not strictly revealed preferred to Korea. Second, the revealed preference rationalizes the choices because the decision maker only chooses undominated alternatives. Indeed, out of the full set of alternatives the decision maker chooses \{$a$, $k$\}, which are the only undominated choices according to their revealed preference. \hfill$\lozenge$
\end{Example}

 In this paper, we provide an alternative formulation of WARP that decouples the two tasks performed by WARP. Our first axiom, \hyperlink{axiomt}{Axiom $\tau$}, guarantees that the revealed preference obtained from a choice correspondence is indeed a preference relation, even when it may not rationalize choice (\autoref{Proposition:tau}). A second axiom, \hyperlink{axiomr}{Axiom $\rho$}, guarantees that the revealed preference, even if not a preference relation, nevertheless rationalizes the choice correspondence (\autoref{Proposition:rho}). Together, \hyperlink{axiomt}{Axiom $\tau$} and \hyperlink{axiomr}{Axiom $\rho$} imply that the revealed preference is an actual preference and rationalizes the choice correspondence (\autoref{theorem:theorem1}). 
 
 Besides providing a novel perspective on WARP, our results may be of independent interest to the literatures on non-rational decision making, such as reference dependence (see, e.g., \cite{ok2015revealed} and the references therein) or non-transitive choices (see, e.g., \cite{bernheim2009beyond} and the references therein).
 
\autoref{Proposition:tau} characterizes decision makers whose choices are rational when restricted to binary choices, even if the revealed preference does not rationalize (all of) their choices; an example of such behavior is reference dependence. Under reference dependence, the ranking of two alternatives may depend on factors other than the alternatives themselves; such other factors are denoted \emph{reference points}. A specific kind of reference points is the other alternatives available for choice. Decision makers that exhibit this form of reference dependence should be fully rational when restricted to binary choices because no other alternatives are available for choice. However, their choices from non-binary sets are affected by reference points and so the revealed preference does not rationalize their choices. 


\label{page-rangel}\autoref{Proposition:rho} characterizes decision makers whose revealed preference rationalizes their choices, even if the revealed preference is not transitive. For instance, faced with consumers with nonstandard choice behavior, \cite{bernheim2009beyond} propose a social planner that makes choices by maximizing a binary relation derived from the consumers choices, denoted \emph{unambiguously chosen over}. Like the revealed preference in \autoref{Example:example2}, unambiguously chosen over is non-transitive, yet it has well-defined maximal elements, so the unambiguously chosen over relation fully rationalizes the planner’s choices.

\paragraph{Organization} The paper proceeds as follows.  Section \ref{section:model} presents the model and Section \ref{section:results} states our axioms and results. Having laid out the necessary notation and established our results, Section \ref{section:axioms} discusses the related literature. 

\section{Model}\label{section:model}

\paragraph{Primitives} Let $X$ be a finite set of alternatives and $P(X)$ denote the set of non-empty subsets of $X$. A \emph{choice correspondence}, $c$, is a mapping from $P(X)$ onto itself, $c:P(X)\rightarrow P(X)$, that satisfies $c(A)\subseteq A$ for all $A\in P(X)$. Because $P(X)$ excludes the empty set, this implies that $c$ is non-empty valued.

\paragraph{Preferences} Given $X$, a binary relation on $X$ is any subset of the Cartesian product $X\times X$. Definitions \ref{def:strict preference} and \ref{def:weak preference} define strict and weak preference, respectively:
\begin{Definition}[Strict Preference]\label{def:strict preference}
A binary relation $\succ$ is a \emph{strict preference} if it satisfies the following:
\begin{enumerate}
\item\label{itm:asymmetry} Asymmetry: For all $x,y\in X$, $x\succ y$ implies $y\nsucc x$,
\item\label{itm:negative-transitive} Negative transitivity: For all $x,y,z\in X$, if $x\nsucc y$ and $y\nsucc z$, then $x\nsucc z$.
\end{enumerate}
\end{Definition}

\begin{Definition}[Weak preference]\label{def:weak preference}
A binary relation $\succsim$ is a \emph{weak preference} if it satisfies the following:
\begin{enumerate}
\item\label{itm:complete} Completeness: For all $x,y\in X$, either $x\succsim y$, $y\succsim x$, or both. If both conditions hold we say $x$ and $y$ are \emph{indifferent}, and we denote this as $x\sim y$.
\item\label{itm:transitive} Transitivity: For all $x,y,z\in X$, if $x\succsim y$ and $y\succsim z$, then $x\succsim z$.
\end{enumerate}
\end{Definition}

From any strict preference relation, $\succ$, we define the associated weak preference, $\succsim$, as follows: $x\succsim y$ if and only if $y\nsucc x$. Clearly, if $\succ$ is asymmetric and negatively transitive, $\succsim$ is complete and transitive. Similarly, given a weak preference, $\succsim$, we define the associated strict preference, $\succ$, as follows: $x\succ y$ if and only if $x\succsim y$, but not $y\succsim x$. Clearly, if $\succsim$ is complete and transitive, $\succ$ is asymmetric and negatively transitive.

\paragraph{\emph{Revealed} preference} Whereas preferences may be defined without explicit reference to a choice correspondence, it is desirable that preferences (which are inherently unobservable) be defined in terms of choices, which in principle are observable. This motivates the definition of a \emph{revealed preference}. 

\begin{Definition}[Revealed preference]\label{def:revealed preference}
Given a choice correspondence, $c$, the strict revealed preference obtained from $c$, denoted $\succ_c$, is defined as follows. For each $x,y\in X$, $x\succ_c y\:\Leftrightarrow\: \{x\}=c(\{x,y\})$. The associated weak revealed preference, denoted $\succsim_c$, is defined as $x\succsim_c y\:\Leftrightarrow\:x\in c(\{x,y\})$. 
\end{Definition}

\label{page-binary-choice}\autoref{def:revealed preference} defines revealed preference by looking at choice from binary sets. In Section \ref{section:axioms}, we discuss our results in the context of other definitions of revealed preference, such as that in \cite{richter1971rational}.

Despite its name, nothing so far guarantees that the revealed preference is a \emph{preference} relation in 
the sense of Definitions \ref{def:strict preference} and \ref{def:weak preference}. Because $c$ is non-empty valued, the weak revealed preference, $\succsim_c$, is always complete. However, additional axioms on $c$ are required to guarantee $\succsim_c$ is transitive. Similarly, the strict revealed preference, $\succ_c$, is always asymmetric: $x \succ_c y \Rightarrow \{x\}=c(\{x,y\}) \Rightarrow \{y\}\neq c(\{x,y\})\Rightarrow y \nsucc_c x$. However, additional axioms are needed to guarantee $\succ_c$ is negatively transitive. 

\paragraph{Choices induced by a binary relation} Given a binary relation, \strictb, we can construct the \strictb-undominated choices as follows: for each $A\in P(X)$, $c(A,\strictb)=\{x\in A:\:(\forall y\in A)\: y \nstrictb x\}$.


\paragraph{Rationalizing choice correspondences} Given a choice correspondence, $c$, we can extract the revealed preference, $\succ_c$, and from $\succ_c$ we can then construct the $\succ_c$-undominated choices defined above. We say that $\succ_c$ \emph{rationalizes} $c$ if this composition of operations recovers the original choice correspondence, $c$. That is, $\succ_c$ \emph{rationalizes} $c$ if $c(\cdot)=c(\cdot, \succ_c)$. 



\section{Results}\label{section:results}



Our results disentangle when a choice correspondence generates a revealed preference that is a preference relation, from when the revealed preference rationalizes choice. We do so by providing two axioms, \hyperlink{axiomt}{$\tau$} and \hyperlink{axiomr}{$\rho$}, on the choice correspondence. \hyperlink{axiomt}{Axiom $\tau$} is akin to a transitivity axiom and holds if and only if the choice correspondence generates a preference (\autoref{Proposition:tau}).  \hyperlink{axiomr}{Axiom $\rho$} rules out possible reference dependence and holds if and only if the revealed preference rationalizes choice (\autoref{Proposition:rho}). Together, Axioms \hyperlink{axiomt}{$\tau$} and \hyperlink{axiomr}{$\rho$} imply that a preference relation exists that rationalizes choices (\autoref{theorem:theorem1}). We first preview the statement of \autoref{theorem:theorem1} and then present and discuss the axioms. \label{page-axioms-theorem}

\begin{theorem}\label{theorem:theorem1}
A choice correspondence $c$ satisfies Axioms \hyperlink{axiomt}{$\tau$} and \hyperlink{axiom}{$\rho$} if and only if $\succ_c$ is a preference relation and $\succ_c$ rationalizes $c$.
\end{theorem}

The proof of Theorem \ref{theorem:theorem1} follows from Propositions \ref{Proposition:tau} and \ref{Proposition:rho}, to which we turn next.

\paragraph{Making revealed preference a preference} The revealed preference might not be a preference relation for two reasons. First, the strict revealed preference, $\succ_c$, might not be transitive, leading to choice cycles. If $x\succ_c y$, $y\succ_c z$, and $z\succ_c x$ for some alternatives $x,y,z\in X$, then $\succ_c$ is not a preference. Moreover, $\succ_c$-undominated choices may be empty-valued--i.e., $c(\{x,y,z\},\succ_c)=\emptyset$--because all alternatives are dominated. Second, $\succ_c$ might not be a preference because it may not be transitive through indifference. That is, if $x\succ_c y$, $y \sim_c z$, and $x\sim_c z$ for some alternatives $x,y,z\in X$, then $\succ_c$ is not a preference. In contrast to the previous example, the $\succ_c$-undominated choices are not empty-valued: if $x\succ_c y$, $y \sim_c z$, and $x\sim_c z$ for some alternatives $x,y,z\in X$, then $c(\{x,y,z\},\succ_c)=\{x,z\}$.

The following axiom rules out these two problems, and is necessary and sufficient for $\succ_c$ to be a strict preference. Because the axiom is essentially a transitivity axiom, we denote it by $\tau$.

\begin{AxiomT}\label{AxiomTau}\hypertarget{axiomt}{}
For all $x,y,z\in X$, if $z\in c(\{z,y\})$ and $\{x\}=c(\{x,z\})$, then $\{x\}=c(\{x,y\})$.
\end{AxiomT}

Clearly, \autoref{Example:example1} satisfies \hyperlink{axiomt}{Axiom $\tau$}, whereas \autoref{Example:example2} does not. In \autoref{Example:example2}, setting $x=a$, $z=b$, and $y=k$, we obtain a violation of \hyperlink{axiomt}{Axiom $\tau$}. By \autoref{Proposition:tau} below, the revealed preference in \autoref{Example:example1} is a preference, but it is not in \autoref{Example:example2}. 

\begin{Proposition}\label{Proposition:tau}\hypertarget{Proposition:tau}
Choice correspondence $c$ satisfies \hyperlink{axiomt}{Axiom $\tau$} if and only if $\succ_c$ is a preference relation.
\end{Proposition}
\begin{proof}
$(\Rightarrow):$ 
Suppose $c$ satisfies Axiom $\tau$.
By construction, $\succ_c$ is always asymmetric.
Now, let $x,y,z$ be such that $x\nsucc_c y$, $y\nsucc_c z$.
By way of contradiction, assume $x\succ_c z$.
Then, we have the following conditions: $\{x\}=c(\{x,z\})$, because $x\succ_c z$; $z\in c(\{z,y\})$, because $y\nsucc_c z$; and $y\in c(\{x,y\})$, because $x\nsucc_c y$.
From $\{x\}=c(\{x,z\})$ and $z\in c(\{z,y\})$, Axiom $\tau$ implies $\{x\}=c(\{x,y\})$.
This is a contradiction because $y\in c(\{x,y\})$.
Thus, $\succ_c$ is negatively transitive and asymmetric.

$(\Leftarrow):$ 
Suppose $\succ_c$ is a (strict) preference relation.
Then, the weak part of $\succ_c$, $\succsim_c$, is transitive.
Let $x,y,z\in X$ satisfy that $\{x\}=c(\{x,z\})$ and $z\in c(\{z,y\})$.
Then, $x \succsim_c z$ and $z\succsim_c y$.
Thus, $x\succsim_c y$, so $x\in c(\{x,y\})$.
It remains to show that $y\notin c(\{x,y\})$.
If $y\in c(\{x,y\})$ then $y\succsim_c x$.
Because we already have that $z\succsim_c y$, transitivity of $\succsim_c$ implies that $z\succsim_c x$.
Thus, $z\in c(\{x,z\})$, contradicting that $\{x\}=c(\{x,z\})$.
Hence, $y\notin c(\{x,y\})$, thus showing that $\{x\}=c(\{x,y\})$, and Axiom $\tau$ holds.
\end{proof}

\paragraph{Revealed preference rationalizes choice} 
\label{page-luce-raiffa}A choice correspondence may not be rationalized by the revealed preference for two reasons. First, the set of undominated choices for a specific choice set might be empty. Second, how favorably an alternative $x$ compares to another alternative $y$ might depend on what other alternatives, $z$, are also available to choose from. To illustrate, recall \citeauthor{luce1989games}'s restaurant example. When offered a choice between chicken ($y$) and steak tartare ($x$), the diner chooses $y$ -- that is, $c(\{x,y\})=\{y\}$. If, in addition, the diner is offered frog legs ($z$), the diner chooses steak tartare-- that is, $c(\{x,y,z\})=\{x\}$. In the language of \cite{ok2015revealed},  $z$ acts as a \emph{revealed reference} for $y$: the availability of $z$ signals the restaurant is of good quality and therefore the diner is willing to choose $x$. Our \hyperlink{axiomr}{Axiom $\rho$} rules out alternatives acting as reference points, and thus guarantees that $c$ can be rationalized by $\succ_c$.


\begin{AxiomR}\label{AxiomRho}\hypertarget{axiomr}
For all $x,y\in X $ and all $B\subset X$, $x\in c(\{x,y\})\cap c(B\cup \{x\})$ if and only if $x\in c(B\cup\{x,y\})$.
\end{AxiomR}


To understand \hyperlink{axiomr}{Axiom $\rho$}, suppose alternative $x$ is chosen from the set $B\cup\{x,y\}$. This could happen because $x$ is judged at least as good as $y$--while also being chosen in $B$--and the decision maker is choosing their most preferred alternative. Alternatively, this could happen because some other alternative, $z\in B$, acts as a reference point: once $z$ is removed from $B$, $x$ is no longer as good as $y$. \hyperlink{axiomr}{Axiom $\rho$} rules this out by requiring that $x$ also be chosen over $y$ when we eliminate all potential reference alternatives--that is, $x$ must also be chosen from $\{x,y\}$-- while still being better than anything in $B$. Conversely, assume that $\{x\}=c(\{x,y\})$ and $x\in c(B\cup\{x\})$. Then, $x$ is revealed preferred to $y$ while also being as good as any other $z\in B$. However, adding $y$ to $B\cup\{x\}$ could act as a reference point for some $z\in B$ so that $x$ is no longer as good as all other alternatives in $B\cup\{x,y\}$; \hyperlink{axiomr}{Axiom $\rho$} rules this out as well. 

Clearly, \autoref{Example:example2} satisfies \hyperlink{axiomr}{Axiom $\rho$}, whereas \autoref{Example:example1} does not. In \autoref{Example:example1}, setting $B=\{b,k\}$, $x=d$, $y=a$, we obtain a violation of \hyperlink{axiomr}{Axiom $\rho$}. By \autoref{Proposition:rho} below, the revealed preference rationalizes choices in \autoref{Example:example2}, but it does not in \autoref{Example:example1}.

\begin{Proposition}\label{Proposition:rho}\hypertarget{Proposition:rho}
Choice correspondence $c$ satisfies \hyperlink{axiomr}{Axiom $\rho$} if and only if $\succ_c$ rationalizes $c$.
\end{Proposition}
\begin{proof}
$(\Leftarrow):$ 
Suppose $\succ_c$ rationalizes $c$.
Choose any $x,y\in X$ and $B\subset X$.

We start by checking that if $x\in c(B\cup\{x,y\})$ then $x\in c(B\cup\{x\})$ and $x\in c(\{x,y\})$.
Because $\succ_c$ rationalizes $c$, then $x\in c(B\cup\{x,y\})$ implies $x\in c(\{x,z \})$ for all $z\in B\cup\{x,y\}$.
Therefore, $x\in c(\{x,y\})$ and $x\in c(B\cup\{x\}, \succ_c)$.
Again, because $\succ_c$ rationalizes $c$, $x\in c(B\cup\{x\}, \succ_c)$ implies $x\in c(B\cup\{x\})$ and this completes the argument. 

We now check the converse: if $x\in c(B\cup\{x\})$ and $x\in c(\{x,y\})$ then $x\in c(B\cup\{x,y\})$.
Because $\succ_c$ rationalizes $c$, then $x\in c(B\cup\{x\})$ implies $x\in c(\{x,z \})$ for all $z\in B$.
Because we already know that $x\in c(\{x,y\})$, we get $x\in c(\{x,z \})$ for all $z\in B\cup\{y\}$.
Because $\succ_c$ rationalizes choice, this implies $x\in c(B\cup\{x,y\})$ and this completes the argument.

$(\Rightarrow):$ 
Suppose Axiom $\rho$ holds, we want to show that $c(A)=c(A,\succ_c)$.
Let $A\in P(X)$.
If $A$ has two elements, then $c(A)=c(A,\succ_c)$ by construction and this concludes the proof.
We now need to show that $c(A)=c(A,\succ_c)$ for sets $A$ that have at least three elements.
We proceed with an iterative argument: consider the predicate p($k$)= ``Let $A\subset X$ have $k\in\{3,...,K\}$ elements. If  $x\in c(A,\succ_c)$, then $x\in c(A)$".
For $k=3$ this claim is immediately true because of Axiom $\rho$.
Indeed, $x\in c(A,\succ_c)$ implies $x\in c(\{x,y\})$ for each $y\in A\equiv\{a_1,a_2,a_3\}$.
Without loss of generality, set $x\equiv a_1$ and choose $y=a_2$ and $B=\{a_3\}$. 
By Axiom $\rho$, $x\in c(\{x,a_2\})\cap c(\{x\}\cup\{a_3\})$ implies $x\in c(\{x,a_2,a_3\})=c(A)$ and this completes the argument.  
Suppose now that p($k$) is true for $k$, we want to show it is true for $k+1$.
Let $A$ have $k+1$ elements and suppose $x\in c(A,\succ_c)$.
Without loss of generality, $A=\{a_1,....,a_{k+1}\}$ with $x=a_1$.
Define $B=\{a_1,...,a_k\}$ and $y=a_{k+1}$.
We make two observations.
First, $x\in c(A,\succ_c)$ implies $x\in c(B,\succ_c)$ and $x\in c(\{x,y\})$ by definition of $c(\cdot,\succ_c)$.
Second, $B$ has $k$ elements.
Thus, p($k$) implies $x\in c(B)$.
Because we also know that $x\in c(\{x,y\})$, Axiom $\rho$ implies $x\in c(B\cup\{x,y\})$. 
Because $B\cup\{x,y\}=A$ this completes the argument.

We now show that $c(A)\subset c(A,\succ_c)$.
Let $x\in c(A)$ and let $y\in A$ be arbitrarily selected.
Then, Axiom $\rho$ implies $x\in c(\{x,y\})\cap c(A)$; simply select $B=A=A\cup\{x,y\}$ in the axiom.
In particular, $x\in c(\{x,y\})$.
Because $y\in A$ was arbitrarily selected this implies that $x\in c(A,\succ_c)$ and this concludes the proof.
\end{proof}

\autoref{Proposition:tau} provides a necessary and sufficient condition for $\succ_c$ to be a preference relation, even when it does not rationalize choice. Similarly, \autoref{Proposition:rho} provides a necessary and sufficient condition for $\succ_c$ to rationalize choice, even when $\succ_c$ is not an actual preference. Thus, these propositions jointly replicate the rationalization result from WARP. Moreover, these axioms disentangle when the revealed preference is indeed a preference from when the revealed preference rationalizes choice.


\section{Related literature}\label{section:axioms} 

\paragraph{Rationalizing choice} \label{page-richter}Classic papers such as \cite{samuelson1938note}, \cite{houthakker1950revealed}, \cite{richter1966revealed}, \cite{richter1971rational}, and \cite{sen1971choice}, provide different characterizations of when choice behavior is consistent with maximizing a preference relation. Among these, the most closely related are \cite{richter1971rational} and \cite{sen1971choice}. 

\label{page-richter2}Both \cite{richter1971rational} and \cite{sen1971choice} work with a binary relation, which Sen dubs the ``at least as good" relation, defined as follows: an alternative $x$ is \emph{at least as good as} an alternative $y$, denoted $xVy$, if a set $A$ in $X$ exists such that $x,y\in A$ and $x\in c(A)$. The $V$-maximal choices out of a set $A$ are given by $c_M(A,V)=\{x\in A: \: (\forall y\in A)\: xVy\: \}$. When $c(\cdot)=c_M(\cdot,V)$, \cite{richter1971rational} says that $c$ satisfies the $V$-axiom, while \cite{sen1971choice} says that $c$ is \emph{normal}. 


\label{verbatim} \autoref{Proposition:rho} shows that $c(\cdot)=c(\cdot,\succ_c)$, that is, $\succ_c$ rationalizes $c$, if and only if it satisfies \hyperlink{axiomr}{Axiom $\rho$}. \cite{richter1971rational}  studies a similar question and shows that $c=c_M(\cdot,\weakb)$ for some reflexive binary relation \weakb\ if and only if  $c=c_M(\cdot,V)$, that is, $c$ satisfies the $V$-axiom.\footnote{We say \weakb\ is \emph{reflexive} if for all $x\in X$, $x\weakb x$.} To connect both exercises, we need to connect rationalization via the strict relation $\succ_c$ with the property $c=c_M(\cdot,V)$. Using results from \cite{sen1971choice}, \autoref{Lemma:normal} shows that $c$ is rationalizable by $\succ_c$ if and only if $c=c_M(\cdot, V)$--equivalently, $c$ satisfies the $V$-axiom.

\begin{Lemma}\label{Lemma:normal}
Choice correspondence $c$ satisfies the $V$-axiom if and only if $\succ_c$ rationalizes $c$.
\end{Lemma}

Consequently, \autoref{Lemma:normal} shows that \autoref{Proposition:rho} also holds if we adopt $V$ as our definition of revealed preference.

\begin{proof}
($\Rightarrow$) If $c$ satisfies the $V$-axiom, then $c$ is normal. Letting $\succsim_c$ denote the weak part of $\succ_c$, \cite{sen1971choice} shows that $\succsim_c=V$. Thus, $c_M(\cdot,V)=c_M(\cdot,\succsim_c)=c(\cdot,\succ_c)$. Because $c$ is normal, $c=c_M(\cdot,V)=c(\cdot,\succ_c)$, so $\succ_c$ rationalizes $c$.

($\Leftarrow$) Suppose $\succ_c$ rationalizes $c$. It suffices to show that $c=c_M(\cdot,V)$. For each $A\subset X$, $x\in c(A)\Rightarrow x\in c(A,\succ_c)\Rightarrow\:(\forall y\in A),\: x\in c(\{x,y\})\Rightarrow\:(\forall y\in A)\: xVy\Rightarrow\:x\in c_M(A,V)$. Therefore, $c(\cdot)\subset c_M(\cdot,V)$. We now show the converse: for all $A\subset X$, $x\in c_M(A,V)\Rightarrow (\forall\:y\in A)\:(\exists B_y) : x,y\in B_y\: \text{and}\:x\in c(B_y)$. Because $\succ_c$ rationalizes $c$ and because $x,y\in B_y$ this implies $x\in c(\{x,y\})$ for all $y\in A$, so that $x\in c(A)$. 
\end{proof}

Assuming normality, \cite{sen1971choice} also studies when $V$ is quasi-transitive, that is, when its strict part is transitive. Under normality,  \cite{sen1971choice} shows that quasi-transitivity of $V$ is equivalent to Axiom $\delta$: $(\forall S,T\in X)$, $(\forall x,y\in c(S))$, if $S\subset T$ then $\{x\}\neq c(T)$. In contrast, \autoref{Proposition:tau} shows that \hyperlink{axiomt}{Axiom $\tau$} is equivalent to negative transitivity of $\succ_c$ independently of whether $c$ is normal.\footnote{Under normality, \hyperlink{axiomt}{Axiom $\tau$} implies Axiom $\delta$: \hyperlink{axiomt}{Axiom $\tau$} implies transitivity of $\succ_c$, which coincides with the strict part of $V$ because of normality. Thus, Axiom $\delta$ holds. However, as \autoref{Example:example2} illustrates, the converse is false.}

%

\paragraph{Reference dependence} \cite{ok2015revealed} define when an alternative $z$ is a reference point for choosing an alternative $x$ over an alternative $y$: either $x \in c(\{x,y,z\})$ but $x\notin c(\{x,y\})$, or $y\in c(\{x,y\})$ but $\{x,y\}\cap c(\{x,y,z\})=\{x\}$. Both these definitions have the flavor that $x$ cannot ``beat" $y$ by itself, but it can do so with the help of $z$. 

Axiom $\rho$ generalizes this idea by ruling out that no \emph{set} of alternatives can be a reference point. Consider a choice set with four alternatives: $X=\{x,y,z,w\}$. Define $c$ so that $\{x\}=c(\{x,y\})=c(\{x,y,z\})=c(\{x,y,w\})$ and $y\in c(\{x,y,w,z\})$. Neither $w$ nor $z$ \emph{individually} act as a reference point that overturns the revealed preference, $x\succ_c y$, because $x$ is still chosen over $y$ from all three element sets. However, adding both $w$ and $z$ \emph{simultaneously} acts as a reference point that makes $y$ be chosen over $x$. 

That sets of alternatives, rather than individual alternatives, may act as reference points is an idea present in other papers as well.
The ``categorize-then-choose" procedure introduced in \cite{categorize-then-choose} is such an example. A decision maker makes choices out of a set $A$ by first dividing $A$ into subsets; these subsets are denoted ``categories". For example, when choosing what to eat, the decision maker might place different dishes into categories--Italian, Indian, Asian, and American. Next, the decision maker ranks these categories--for instance, Italian might dominate all other categories. Finally, the decision maker ignores all alternatives that belong to a dominated category and chooses their most preferred alternative from those that remain. It is immediate to check that this procedure need not satisfy Axiom $\rho$, so the revealed preference need not rationalize choice. In our example, the decision maker's preferred food might be an American dish, but because this alternative belongs to a dominated category, it is never chosen. In this example, it is the categories, not any individual alternative, that act as reference sets that distort decision making away from $c(\cdot,\succ_c)$. 

\section{Conclusions}\label{section:conclusions}
 As the introductory examples illustrate, the revealed preference may be a well-defined preference even if it does not rationalize choice, and the revealed preference may fail to be an actual preference even when it does rationalize choice. We shed new light on a classical result--that WARP holds if and only if the revealed preference is a preference relation that rationalizes choice--by disentangling whether the revealed preference is an actual preference from whether it rationalizes choice. 
\bibliographystyle{ecta}
\bibliography{WARPbib} 
\end{document}